%% file: main.tex
\documentclass[11pt]{amsart}
\usepackage{braket,amsfonts,amssymb,mathtools,stmaryrd, graphicx,a4,bbm,mdframed,mathrsfs,hyperref}
\graphicspath{{./figures/}}
\usepackage{float}
\usepackage{tikz, tkz-euclide}
\usepackage{tikz-cd}
\usepackage[title]{appendix}
\usetikzlibrary{calc}
\usetikzlibrary{arrows}
\usetikzlibrary{positioning}
\usetikzlibrary{decorations.markings}
\usetikzlibrary{tikzmark}
\usetikzlibrary{tikzmark}
\usepackage[normalem]{ulem}
\usepackage{euscript} 
\usepackage{color}\definecolor{Red}{cmyk}{0,1,1,0}
\usepackage{enumitem}

\newcommand{\spann}{span}

\newcommand{\Var}[0]{\text{Var}}
\newcommand\redsout{\bgroup\markoverwith{\textcolor{red}{\rule[0.5ex]{2pt}{0.4pt}}}\ULon}

\newtheorem{theorem}           {Theorem}

\theoremstyle{definition}

\newtheorem*{theorem*}{Theorem}
\newtheorem*{conjecture*}   {Conjecture}
\newtheorem*{corollary*}   {Corollary}
  
  \newtheorem*{lemma*}          {Lemma}
    \newtheorem*{claim*}          {Claim}
  \newtheorem{definition}[theorem]         {Definition}
  
  \newtheorem{proposition}[theorem]      {Proposition}

  \theoremstyle{definition}

\newcommand{\partiald}[2]{\frac{\partial {#1}}{\partial {#2}}}
\DeclareMathOperator{\sgn}{sgn}
\DeclareMathOperator{\Image}{Im}
\DeclareMathOperator{\SRot}{SRot}

\begin{document}


\voffset=-1.5truecm\hsize=16.5truecm    \vsize=24.truecm
\baselineskip=14pt plus0.1pt minus0.1pt \parindent=12pt
\lineskip=4pt\lineskiplimit=0.1pt      \parskip=0.1pt plus1pt

\def\ds{\displaystyle}\def\st{\scriptstyle}\def\sst{\scriptscriptstyle}



\global\newcount\numsec\global\newcount\numfor
\gdef\profonditastruttura{\dp\strutbox}
\def\senondefinito#1{\expandafter\ifx\csname#1\endcsname\relax}
\def\SIA #1,#2,#3 {\senondefinito{#1#2}
\expandafter\xdef\csname #1#2\endcsname{#3} \else
\write16{???? il simbolo #2 e' gia' stato definito !!!!} \fi}
\def\etichetta(#1){(\veroparagrafo.\veraformula)
\SIA e,#1,(\veroparagrafo.\veraformula)
 \global\advance\numfor by 1
 \write16{ EQ \equ(#1) ha simbolo #1 }}
\def\etichettaa(#1){(A\veroparagrafo.\veraformula)
 \SIA e,#1,(A\veroparagrafo.\veraformula)
 \global\advance\numfor by 1\write16{ EQ \equ(#1) ha simbolo #1 }}
\def\BOZZA{\def\alato(##1){
 {\vtop to \profonditastruttura{\baselineskip
 \profonditastruttura\vss
 \rlap{\kern-\hsize\kern-1.2truecm{$\scriptstyle##1$}}}}}}
\def\alato(#1){}
\def\veroparagrafo{\number\numsec}\def\veraformula{\number\numfor}
\def\Eq(#1){\eqno{\etichetta(#1)\alato(#1)}}
\def\eq(#1){\etichetta(#1)\alato(#1)}
\def\Eqa(#1){\eqno{\etichettaa(#1)\alato(#1)}}
\def\eqa(#1){\etichettaa(#1)\alato(#1)}
\def\equ(#1){\senondefinito{e#1}$\clubsuit$#1\else\csname e#1\endcsname\fi}
\let\EQ=\Eq


\def\\{\noindent}
\let\io=\infty

\def\VU{{\mathbb{V}}}
\def\EE{{\mathbb{E}}}
\def\GI{{\mathbb{G}}}
\def\TT{{\mathbb{T}}}
\def\C{\mathbb{C}}
\def\LL{{\cal L}}
\def\RR{{\cal R}}
\def\SS{{\cal S}}
\def\NN{{\cal N}}
\def\HH{{\cal H}}
\def\GG{{\cal G}}
\def\PP{{\cal P}}
\def\AA{{\cal A}}
\def\BB{{\cal B}}
\def\FF{{\cal F}}
\def\vv{\vskip.2cm}
\def\gt{{\tilde\g}}
\def\E{{\mathcal E} }
\def\I{{\rm I}}

\def\cal{\mathcal}

\def\tende#1{\vtop{\ialign{##\crcr\rightarrowfill\crcr
              \noalign{\kern-1pt\nointerlineskip}
              \hskip3.pt${\scriptstyle #1}$\hskip3.pt\crcr}}}
\def\otto{{\kern-1.truept\leftarrow\kern-5.truept\to\kern-1.truept}}
\def\arm{{}}
\font\bigfnt=cmbx10 scaled\magstep1

\newcommand{\card}[1]{\left|#1\right|}
\newcommand{\und}[1]{\underline{#1}}
\newcommand{\Dom}[0]{\text{Dom}}
\def\1{\rlap{\mbox{\small\rm 1}}\kern.15em 1}
\def\ind#1{\1_{\{#1\}}}
\def\bydef{:=}
\def\defby{=:}
\def\buildd#1#2{\mathrel{\mathop{\kern 0pt#1}\limits_{#2}}}
\def\card#1{\left|#1\right|}
\def\proofof#1{\noindent{\bf Proof of #1. }}
\def\trp{\mathbb{T}}
\def\trt{\mathcal{T}}

\def\bfz{\boldsymbol z}
\def\bfa{\boldsymbol a}
\def\bfalpha{\boldsymbol\alpha}
\def\bfmu{\boldsymbol \mu}
\def\bfmust{\bfT^\infty(\bfmu)}
\def\bfmupr{\boldsymbol {\widetilde\mu}}
\def\bfrho{\boldsymbol \rho}
\def\bfrhost{\boldsymbol \rho^*}
\def\bfrhopr{\boldsymbol {\widetilde\rho}}
\def\bfT{{\boldsymbol T}_{\!\!\bfrho}}
\def\bfR{\boldsymbol R}
\def\bfvarphi{\boldsymbol \varphi}
\def\bfvarphist{\boldsymbol \varphi^*}
\def\bfPi{\boldsymbol \Pi}
\def\bfzero{\boldsymbol 0}
\def\bfW{\boldsymbol W}
\def\formal{\stackrel{\rm F}{=}{}}
\def\eee{{\rm e}}
\def\nnn{\mathcal N}
\def\nst{\nnn^*}
\def\Var{\text{Var}}

\thispagestyle{empty}

\begin{center}
{\LARGE Localizability in de Sitter space for 2+1 dimensions.}
\vskip.5cm
Thiago Raszeja$^{1,2}$, Jo\~ao C. A. Barata$^{1}$ 
\vskip.3cm
\begin{footnotesize}
$^{1}$Institute of Physics of the University of S\~ao Paulo (IFUSP), University of S\~{a}o Paulo, Brazil\\
$^{2}$ Faculty of Applied Mathematics, AGH University of Krakow, Poland
\end{footnotesize}
\vskip.1cm
\begin{scriptsize}
emails: tcraszeja@gmail.com; jbarata@if.usp.br
\end{scriptsize}

\end{center}

\def\be{\begin{equation}}
\def\ee{\end{equation}}

\vskip1.0cm
\begin{quote}
{\small

\textbf{Abstract.} \begin{footnotesize} We extended the notion of Newton-Wigner localization, already constructed in the bi-dimensional de Sitter space, to the tri-dimensional case for both principal and complementary series. We identify the one-particle subspace, generated by the positive-energy modes solution of the Klein-Gordon equation, as an irreducible representation of the de Sitter group. The time-evolution of the Newton-Wigner operator was obtained explicitly, and for the complementary series the evolution is trivial, i.e., there are no dynamics. Also, we discussed heuristically the existing sign ambiguity when we do not require as a postulate that the Newton-Wigner functions must be proportional to their respective solutions in the representation of solutions of the Klein-Gordon equation.
\end{footnotesize}

}
\end{quote}

\numsec=2\numfor=1
\section*{Introduction}
\vv
\noindent
\input{0_introduction}

\numsec=2\numfor=1

\section{The Klein-Gordon equation on the tri-dimensional de Sitter space and its solutions} \label{sec:klein_gordon}

\input{1_klein_gordon}

\section{One-particle space, canonical quantization, and symmetries} \label{sec:quantization}

\input{2_quantization}

\section{Localizability} \label{sec:localizability}
\input{3_localizability}

\section{Concluding remarks}

This paper extends to the tri-dimensional case of the postulates of localizability presented in \cite{Barata2012}, and we determined uniquely the family of unitary transformations that realizes the determination of a position operator. Moreover, we explicitly presented this operator. Our results open the possibility of extending even further the theory to a general formulation of the localization postulates for $n+1$-dimensional de Sitter space.

\section*{Acknowledgements}
TR is supported by CNPq and NCN (National Science Center, Poland), Grant 2019/35/D/ST1/01375. We thank N. Yokomizo for fruitful discussions during the development of this work. In addition, TR thanks R. Correa da Silva and J. F. Thuorst for further discussions and friendship.

\bibliographystyle{alpha}
\bibliography{bibliografia}

\end{document}

%% file: 0_introduction.tex
Although it is one of the most fundamental intuitive concepts in physics, the measurement of the position of particles has no obvious mathematical approach in the realm of relativistic quantum mechanics, and it has been developed by many authors over decades \cite{Newton1949, Wightman1962, Philips1964, Kaiser1978, Barata2012}. One of the central topics in quantum field theory (QFT) concerns the interaction among particles and collision processes, where the notion of position is required since the typical experiments are collision experiments measured via detectors, realized in a limited region \cite{Haag1965, Buchholz1986, Schroer2010, Araki1967}. Moreover, position operators are natural objects when we take into account principles that connect classical and quantum systems via classical limit \cite{Hepp1974, Barata2009}.

On the other hand, the concern about creating a unified theory that covers physical phenomena in totality, QFT on curved space-time (QFTCS), as a semi-classical gravitation theory, was responsible for impressive forethoughts, such as the Hawking radiation \cite{Hawking1974,Unruh1981} and the Unruh effect \cite{Fulling1973,Davies1975,Unruh1976}. In particular, the importance of the geometric background in this paper is the following: evidence on the current Universe expansion \cite{Riess1998,Perlmutter1999} and the existence of an inflationary era in the primitive Universe \cite{weinberg2008cosmology,Linde1990,liddle2000} points out a compatibility of the geometry of the Universe with a piece of a de Sitter space. Dynamics on particles in this geometry have been studied in both classical and quantum contexts \cite{Cacciatori2008,Bros2007,Aldrovandi2004}.

Historically, the first relativistic solution for localizability was constructed by T. D. Newton e E. P. Wigner in the Minkowski space \cite{Newton1949}, later revisited in a more rigorous approach in \cite{Wightman1962}, via the establishment of localizability postulates that considers natural invariance arguments in physics, giving a unique position operator, however non-covariant. This problem was solved later in \cite{Philips1964,Kaiser1978}. In a previous work, N. Yokomizo and one of the authors constructed the postulates for the Newton-Wigner (NW) localization on the bi-dimensional de Sitter space for a neutral massive scalar field, leading to a unique solution for a position operator which is well-behaved under the action of the symmetry group $O(2,1)$.

Following similar guidelines of \cite{Barata2012}, we extended the notion of localizability for the tri-dimensional de Sitter space by formulating postulates that carry conditions concerning rotations, parity, and time-reversal operators, and a condition on the large mass limits. These postulates provide a unique family of unitary transformations that carries the Hilbert space of solutions corresponding to the positive energy modes of the Klein-Godon equation to the square-integrable functions on the $2$-sphere. Moreover, we also provided a study on the ambiguity of signals when we do not take into consideration the large mass regularity (postulate IV). An analog ambiguity already occurs in the bi-dimensional case. 

The paper is organized in three main sections.

In section \ref{sec:klein_gordon}, through a convenient choice of coordinates, we describe the Klein-Gordon equation for the tri-dimensional de Sitter space and we solve it. We choose a basis of solutions that separates the positive and the negative modes of energy, and we prove that those solutions have smooth Cauchy conditions, and hence they are well-defined to be in fact a basis for the space.

Section \ref{sec:quantization} presents the quantization of the one-particle Klein-Gordon free field for the positive energy modes and we prove the invariance under $O(3,1)$ of the Hilbert spaces corresponding to the completion of the space generated by these modes. We show explicitly how the parity, time-reversal, main spatial rotations and Lorentz boosts act on each positive mode.

Finally, section \ref{sec:localizability} presents the localizability postulates for a Klein-Gordon for free fields on the de Sitter space $dS^3$, and from them we construct the Newton-Wigner transformations and we prove that it is unique. We also discuss the ambiguity of signals we obtain when we exclude postulate IV and the physical meaning of the choice of these signals.

A substantial amount of the results of this paper are in the Master thesis \cite{Raszeja2015}.

%% file: 1_klein_gordon.tex
\subsection{Local properties of the de Sitter space and the Klein-Gordon equation} The \emph{tri-dimensional de Sitter space} is the Lorentzian submanifold $dS^3 \subset M^4$, where $M^4$ is the $3+1$-dimensional Minkowski space endowed with the metric signature $\eta_{ab} = \text{diag}(1, -1,\cdots, -1)$, and given by 
\begin{equation}
\label{eq:deSitter}
 dS^3 := \{\mathbf{X} \in M^{3+1} : \mathbf{X}^2 = X^a X^b \eta_{ab} = -\alpha^2 \},
\end{equation}
where $\alpha > 0$ is the \emph{de Sitter radius} and the Einstein notation is adopted above. The space $dS^3$ is a hyperboloid of one sheet, whose topology is $S^2 \times \mathbb{R}$, where $S^2$ is the $2$-sphere (spatial part), and its correspondent symmetry group is $O(3,1)$. We adopt in this space the spherical coordinates
\begin{eqnarray}
X^0 &= &\alpha \sinh (t/\alpha); \nonumber \\
X^1 &= &\alpha \cosh (t/\alpha) \sin \theta \cos \varphi; \nonumber \\
X^2 &= &\alpha \cosh (t/\alpha) \sin \theta \sin \varphi; \nonumber \\
X^3 &= &\alpha \cosh (t/\alpha) \cos \theta; \nonumber
\end{eqnarray}
where $t \in \mathbb{R}$, $\theta \in [0, \pi]$ and $ \varphi \in [0, 2 \pi]$. The metric tensor $g$ of the aforementioned choice of coordinates is diagonal, and its non-zero entries are
\begin{equation}
 g_{00} = 1 \text{, } g_{11} = -\alpha^2 \cosh^2(t/\alpha) \text{, } g_{22} = -\alpha^2 \cosh^2(t/\alpha)\sin^2\theta.
\end{equation}
Its volume density is $D_g = \alpha^2 \cosh^2(t/\alpha) \sin \theta$ and hence the D'Alembertian becomes 
\begin{equation*}
\Box = \partial_{tt} + \frac{2}{\alpha} \tanh (t/\alpha) \partial_t - \frac{1}{\alpha^2\cosh^2(t/\alpha)} \left[ \frac{1}{\sin \theta}\partial_{\theta}(\sin \theta \partial_{\theta}) + \frac{1}{\sin^2 \theta} \partial_{\varphi \varphi} \right].
\end{equation*}
Given a submanifold $\Omega \subset M^4$ and $\phi \in C^2(\Omega, \mathbb{C})$, the \emph{Klein-Gordon (KG) equation} is given by
\begin{equation}
\label{eq:KG}
\left( \Box + \frac{m_p^2 + \xi R}{\hslash^2} \right) \phi = 0,
\end{equation}
where $m_p$ is the particle mass, $R$ is the scalar curvatura and $\xi$ is a parameter. In our context, we have $R = \frac{6}{\alpha^2}$. By setting $\mu^2 := m_p^2 + \xi R$, equation \eqref{eq:KG} becomes
\begin{equation}
\label{eq:KG_final}
\left( \Box + \frac{\mu^2}{\hslash^2} \right) \phi = 0.
\end{equation}

\subsection{Solutions of the KG equation on the tri-dimensional de Sitter space: energy modes and their properties}

Now, we solve the KG equation via the separation of variables and find a basis for the space of solutions. By inserting the ansatz
\begin{equation}
\label{eq:ansatz}
\phi (t, \theta, \varphi) = \frac{T(x)}{\sqrt[4]{x^2 - 1}}Y(\theta, \varphi),
\end{equation}
where $x = i \sinh (t/\alpha)$, in \eqref{eq:KG_final}, we obtain the following system of differential equations:
\begin{subequations}
\begin{align}
(1-x^2)T''(x) - 2xT'(x) + \left[ \left( \frac{3}{4} - M^2 \right) - \frac{l(l+1) + 1/4}{1-x^2} \right]T(x) = 0; \label{eq:Legendre} \\
\frac{1}{\sin \theta} \partial_\theta (\sin \theta \partial_\theta Y(\theta, \varphi)) + \frac{1}{\sin^2 \theta} \partial_{\varphi \varphi} Y(\theta, \varphi) + l(l+1)Y(\theta, \varphi) = 0, \label{eq:Laplace2}
\end{align}
\end{subequations}
where $M=\frac{\alpha \mu}{\hslash}$. Equation \eqref{eq:Laplace2} is the angular part of the Laplace equation. Function $Y$ must be continuous, and hence periodic on $S^2$, where we get $l \in \mathbb{N}_0 = \mathbb{N}\cup\{0\}$, then a basis of solutions for \eqref{eq:Laplace2} are the spherical harmonics $Y_l^m$. On the other hand, by writing $\nu(\nu + 1) = \left( \frac{3}{4} - M^2 \right)$ and $\lambda^2 = l(l+1) + \frac{1}{4} = (l+\frac{1}{2})^2$, equation \eqref{eq:Legendre} is the associated Legendre equation,
\begin{equation}
\label{eq:Legendre_geral}
(1-x^2)T''(x) -2xT'(x)+ \left[ \nu(\nu+1) - \frac{\lambda^2}{1-x^2}\right]T(x) = 0,
\end{equation}
where we have
\begin{equation}\label{eq:nu}
\nu = -\frac{1}{2} \pm \sqrt{1 - M^2} \; \text{ and } \; \lambda = \pm \left( l + \frac{1}{2} \right).
\end{equation}
We observe that, for $M^2 \geq 0$ ($\mu^2 \geq 0$), $\nu \in \mathbb{R}$ when $M \in [-1, 1]$, and $\nu = -\frac{1}{2} \pm i\sqrt{M^2-1}$ otherwise. Now, if $M^2 < 0$, then necessarily $\nu \in \mathbb{R}$. In this paper we only consider the case $M^2 >0$, that is, we assume that there is no large negative coupling with the Ricci scalar.

The solutions of \eqref{eq:Legendre_geral} are the \emph{associated Legendre functions} $P_\nu^\lambda$ and $Q_\nu^\lambda$. Observe that these functions are analytic in the whole plane except in the cut in $(-\infty,1]$ (for non-integer indices). We conveniently adapt these solutions via Ferrer's notation and its properties \cite{snow1961}, presented next.

\begin{definition}
\label{def:Ferrer}
Let $\sigma, \nu \in \mathbb{C}$. The function $T_\nu^\sigma:\{z \in \mathbb{C}: |z-1| < 2\} \setminus [1,3) \to \mathbb{C}$, given by
\begin{equation}
\label{eq:TP}
T_\nu^\sigma (z) := e^{\mp i \frac{\sigma \pi}{2}} P_\nu^\sigma (z) \text{, } \pm \Image(z) > 0
\end{equation}
are called \emph{associated Legendre function in Ferrer's notation of degree $\nu$ and order $\sigma$}, or simply \emph{Ferrer functions}.
\end{definition}

The Ferrer functions are also solutions for \eqref{eq:Legendre_geral} and analytic in the whole plane except by its cuts, namely $(-\infty,-1] \cup [1, \infty)$. We choose the solutions
\begin{equation*}
T_\nu^{l+\frac{1}{2}} (i\sinh (t/\alpha)) \text{ and } T_\nu^{l+\frac{1}{2}} (-i\sinh (t/\alpha)).
\end{equation*}
Therefore, a basis of solutions for \eqref{eq:KG_final} consists into the elements
\begin{eqnarray}
u_{l,m}(t, \theta, \varphi) = \sqrt{\frac{\gamma_l}{2 \alpha}} N(\nu)\frac{T_\nu^{l+\frac{1}{2}} (i\sinh (t/\alpha))}{\sqrt{\cosh(t/\alpha)}} Y_l^m(\theta, \varphi); \label{eq:modo_positivo}\\
v_{l,m}(t, \theta, \varphi) = \sqrt{\frac{\gamma_l}{2 \alpha}} N(\nu)\frac{T_\nu^{l+\frac{1}{2}} (-i\sinh (t/\alpha))}{\sqrt{\cosh(t/\alpha)}} Y_l^m(\theta, \varphi), \label{eq:modo_negativo}
\end{eqnarray}
where $N(x) = \mathbbm{1}_{\mathbb{R}} (x) - i \mathbbm{1}_{\{z \in \mathbb{C}: \Image(z) > 0\}}(x) + i \mathbbm{1}_{\{z \in \mathbb{C}: \Image(z) < 0\}}(x)$, where $\mathbbm{1}_A$ is the characteristic function on the set $A$, is a phase, and $\gamma_l = \Gamma(-\nu - l -\frac{1}{2}) \Gamma(\nu - l +\frac{1}{2})$ is a normalization factor. In addition, $l$, $m \in \mathbb{N}_0$ and $-l \leq m \leq l$. The solutions \eqref{eq:modo_positivo} and \eqref{eq:modo_negativo} are interpreted as the \emph{positive} and \emph{negative energy modes} of the Klein-Gordon equation \eqref{eq:KG_final} respectively. We denote by $\mathcal{S}$ the whole space of the solutions of the KG equation, and by $\mathcal{S^+}$ and $\mathcal{S^-}$ the respective subspaces of positive and negative modes, all of them under smooth Cauchy conditions.

\subsection{Regularity properties of the energy modes}

The space complex space $\mathcal{S}$ under smooth Cauchy conditions is given by the vectors $\phi$,
\begin{equation}
\label{eq:vetor_de_s}
\phi (t, \theta, \varphi) = \sum_{l=0}^\infty \sum_{m=-l}^l
\left[ c_{l,m}u_{l,m} + d_{l,m}v_{l,m}\right],
\end{equation}
where the coefficients $c_{l,m}$ satisfy
\begin{equation}
\label{eq:decaimento}
\sum_{l=0}^\infty \sum_{m=-l}^l |c_{l,m}|(l^2 + m^2)^{j/2} < \infty \text{, } \forall j > 0;
\end{equation}
and $d_{l,m}$ satisfies the above by mutatis mutandis. In this subsection, we show that those solutions are well-defined in $\mathcal{S}$, by proving that the series \eqref{eq:vetor_de_s} converges absolute and uniformly to solutions \eqref{eq:KG_final} with smooth initial conditions.

We recall that all the metrics for a finite-dimensional vector space are equivalent, and since the set of indexes $\{(l,m) \in \mathbb{Z}^2: l\geq 0, -l \leq m \leq l \}$ is contained in $\mathbb{R}^2$, there are $R_1$, $R_2 > 0$ s.t.
\begin{align}
R_1 \sum_{l=0}^\infty \sum_{m=-l}^l |c_{l,m}| (\max\{|l|,|m|\})^j \leq \sum_{l=0}^\infty \sum_{m=-l}^l |c_{l,m}|(l^2 + m^2)^{j/2}  \leq  R_2 \sum_{l=0}^\infty \sum_{m=-l}^l |c_{l,m}| (\max\{|l|,|m|\})^j,\label{eq:metrica_equiv}
\end{align}
where we conclude that
\begin{equation}
\label{eq:metrica_sup}
\sum_{l=0}^\infty \sum_{m=-l}^l |c_{l,m}|(l^2 + m^2)^{j/2} \leq R_2 \sum_{l=0}^\infty \sum_{m=-l}^l |c_{l,m}|l^j < \infty .
\end{equation}
Then, the decaying dependence of the coefficients can be written in a simpler way, and straightforward to compare with the bi-dimensional de Sitter case \cite{Barata2012}. Considering only the terms $c_{l,m}$, let the spherical harmonics expansion at fixed time $t$
\begin{equation}
\label{eq:expansao}
H(t, \theta, \varphi) :=  \sum_{l=0}^{\infty}\sum_{m=-l}^{l}p_{l,m}Y_l^m(\theta, \varphi), \: p_{l,m} := c_{l,m} \sqrt{\frac{\gamma_l}{2 \alpha}} N(\nu) \frac{T_\nu^{l+\frac{1}{2}} (i\sinh (t/\alpha))}{\sqrt{\cosh(t/\alpha)}}.
\end{equation}
The asymptotic behavior of the coefficients $p_{l,m}$ is obtained by the identity (see \cite{snow1961})
\begin{equation}
\left|\sqrt{\frac{\gamma_l}{2 \alpha}} N(\nu) T_\nu^{l+\frac{1}{2}} (z)\right| \simeq \frac{1}{\sqrt{2\alpha}} \left|\frac{\Gamma\left(- \nu -l + \frac{1}{2} \right)}{\Gamma\left( \nu -l + \frac{1}{2} \right)} \right|^{\frac{1}{2}}\left|\Gamma\left( \nu -l + \frac{1}{2}  \right)T_\nu^{l+\frac{1}{2}} (z)\right| \simeq \frac{1}{\sqrt{2 \alpha}} l^{-\frac{1}{2}\pm \chi},
\end{equation}
where $0 \leq \chi \leq 1$, and hence,
\begin{equation}
|p_{l,m}| = |c_{l,m}| \frac{l^{-\frac{1}{2}\pm \chi}}{\sqrt{2\alpha \cosh(t/\alpha)}},
\end{equation}
and the rapid decay of $c_{l,m}$, the absolute and uniform convergence in \eqref{eq:expansao} is ensured, and the same arguments presented above also follows for $d_{l,m}$. Moreover, the terms in \eqref{eq:vetor_de_s} are continuous and then $\phi(t,\theta, \varphi)$ is continuous as well.

Now we analyze the derivatives in the form $\partial_\varphi^n \partial_\theta^q \partial_t^2 \phi$, $\{n, q\} \subset \mathbb{N}_0$. Note that $\partial_\varphi^n$ simply gives a factor $(im)^n$, so it does not affect the continuity. On the other hand, $\partial_\theta^q$ acts on the associated Legendre polynomials only, and it is straightforward that the continuity is preserved. It remains to study how $\partial_t^2$ affects $\phi$, i.e., how it acts on $T(t) = \frac{T_\nu^{l+\frac{1}{2}} (i\sinh (t/\alpha))}{\sqrt{\cosh(t/\alpha)}}$. Since $T$ satisfies
\begin{equation}
\label{eq:sem_ansatz}
(1-x^2) \frac{d^2 T}{dx^2} (x) - 3x\frac{d T}{dx}(x) - \left[M^2 + \frac{l(l+1/2)}{1-x^2} \right] T(x) = 0, \text{ } x = i \sinh (t/\alpha),
\end{equation}
and by recurrence relations of the Ferrer functions \cite{snow1961}, in the asymptoptic limit on $l$, one gets
\begin{eqnarray}
\frac{d}{dt} \left(\frac{T_\nu^{l+\frac{1}{2}} (i\sinh (t/\alpha))}{\sqrt{\cosh(t/\alpha)}} \right) \simeq \frac{il^2}{2 \alpha } \frac{T^{l+\frac{1}{2}}_\nu (i \sinh (t/\alpha))}{\sqrt{\cosh(t/\alpha)}},
\end{eqnarray}
\begin{equation}
\label{eq:segunda_derivada}
\frac{d^2}{dt^2} \left(\frac{T_\nu^{l+\frac{1}{2}} (i\sinh (t/\alpha))}{\sqrt{\cosh(t/\alpha)}} \right) \simeq - \frac{l^2}{\alpha^2} \left[ i \frac{\sinh(t/\alpha)}{\cosh (t/\alpha)} + \frac{1}{\cosh^2 (t/\alpha)}   \right] \frac{T^{l+\frac{1}{2}}(i \sinh (t/\alpha))}{\sqrt{\cosh(t/\alpha)}}.
\end{equation}
Therefore, the time derivatives of first and second order include a factor of polynomial order in $l$, and once again the rapid decay of $c_{l,m}$ ensures the convergence of such a derivatives acting on \eqref{eq:expansao}, such results also holds for $d_{l,m}$. Therefore, $\phi$ is a $C^2$-function, and via superposition it is also a solution of the KG equation. Furthermore, for a fixed time $t$, $\phi$ is a smooth function on the $2$-sphere, since any degree of combination of derivatives on $\varphi$ and $\theta$ acting on  \eqref{eq:expansao} preserves its regularity, and by \eqref{eq:segunda_derivada}, the result also holds for $\partial_t \phi$. Therefore, from this point we only consider the smooth Cauchy conditions on $\phi(0,\theta,\varphi)$ and $\partial_t \phi (0,\theta,\varphi)$. Since every function in $C^0(S^2)$ has a convergent Laplace expansion, then $\mathcal{S}$ has every solution with smooth initial coefficients.

The space $\mathcal{S}$ is endowed with the usual Hermitian sesquilinear form in QFTCS, namely 
\begin{equation}
\label{eq:prod_int}
 \braket{f,g} = i [a(t)]^2 \int_{S_t} d\varphi d\theta [(f^*\overleftrightarrow{\partial_t}g)\sin \theta],
\end{equation}
where $S_t$ is a time slice of constant $t$, and $a(t):=\alpha \cosh (t/\alpha)$ is the radius of $S_t$. Observe that $\braket{f^*,g^*} = -\braket{f,g}$. Proposition \ref{prop:prod_int} shows that we can separate $\mathcal{S}$ into the subspaces of positive $\mathcal{S}^+$ and negative $\mathcal{S}^-$ energies for $M \neq 1$ (the case $M = 1$ will not be treated in this work).

\begin{proposition}
\label{prop:prod_int}
The solutions \eqref{eq:modo_positivo} and \eqref{eq:modo_negativo} satisfy
\begin{equation}
\label{eq:deltas}
\braket{u_{l,m},u_{p,s}} = \delta_{lp} \delta_{ms}, \quad  \braket{v_{l,m},v_{p,s}} = - \delta_{lp} \delta_{ms} \quad \text{e} \quad \braket{u_{l,m},v_{p,s}} = 0
\end{equation}
for every $M \neq 1$.
\end{proposition}

\begin{proof}
Let $z = i\sinh (t/\alpha)$. Then, $\partial_t = \frac{i}{\alpha} \sqrt{1-z^2} \partial_z$. By the usual orthogonality relations on spherical harmonics, we have
\begin{equation*}
\braket{u_{l,m},u_{p,s}} = - (1-z^2)^{\frac{3}{2}} \frac{|\gamma_l|}{2} \left[ \left(\frac{T_\nu^{l+\frac{1}{2}} (z)}{\sqrt[4]{1-z^2}}\right)^*\overleftrightarrow{\partial_z}\left(\frac{T_\nu^{l+\frac{1}{2}} (z)}{\sqrt[4]{1-z^2}} \right) \right] \delta_{lp} \delta_{ms}.
\end{equation*}
If $M<1$, then one can show that $[T_\nu^{l+\frac{1}{2}} (z)]^* = T_\nu^{l+\frac{1}{2}} (-z)$, and hence
\begin{equation}
\label{eq:ortog_parcial}
\braket{u_{l,m},u_{p,s}} = - (1-z^2)^{\frac{3}{2}} \frac{|\gamma_l|}{2} \left[T_\nu^{l+\frac{1}{2}} (-z) \overleftrightarrow{\partial_z}  T_\nu^{l+\frac{1}{2}} (z)     \right].
\end{equation}
By another property of the Ferrer functions, namely
\begin{equation}
\label{eq:comut_deriv}
\small
(1-z^2)\left[ T_\nu^\sigma (z) \frac{d}{dz} T_\nu^\sigma (-z) - T_\nu^\sigma (-z) \frac{d}{dz} T_\nu^\sigma (z)\right] = \frac{2 \sin [(\nu - \sigma)\pi]}{\sin [(\nu + \sigma)\pi]\Gamma(-\nu-\sigma)\Gamma(\nu-\sigma+1)},
\end{equation}
we obtain
\begin{equation}
\label{eq:M<1}
\braket{u_{l,m},u_{p,s}} = - \frac{|\gamma_l|}{\gamma_l} \delta_{lp} \delta_{ms}.
\end{equation}
Now, if $M>1$, one can show that $[T_\nu^{l+\frac{1}{2}} (z)]^* =- T_\nu^{l+\frac{1}{2}} (-z)$ and then
\begin{equation}
\label{eq:M>1}
\braket{u_{l,m},u_{p,s}} = \frac{|\gamma_l|}{\gamma_l} \delta_{lp} \delta_{ms}.
\end{equation}
Now we analyze the signal of $\gamma_l = (-\nu-l-\frac{1}{2})$ for each case. By the Gamma function recurrence property, one gets
\begin{equation*}
\frac{1}{\gamma_l} = \frac{\prod_{p=0}^l\left[\left(-\nu-p-\frac{1}{2}\right) \left(\nu-p+\frac{1}{2}\right)\right]}{\Gamma \left(-\nu -\frac{1}{2}\right)\left(\nu + \frac{1}{2}\right)} = \frac{1}{\gamma_0}\underbrace{\prod_{p=1}^l \left( |1-M^2| +p^2\right)}_{>0}.
\end{equation*}
Then $\gamma_l$ and $\gamma_0$ have the same signal. The latter one satisfy
\begin{equation}
\label{eq:gamma_zero_final}
\frac{1}{\gamma_0} = -\left(\nu + \frac{1}{2}\right)\frac{\cos (\nu \pi)}{\pi}.
\end{equation}
Now, for $M<1$, we necessarily have $\nu \in \left( -\frac{3}{2},-\frac{1}{2} \right) \cup \left( -\frac{1}{2},\frac{1}{2} \right)$. If $\nu \in \left(-\frac{3}{2},-\frac{1}{2}\right)$, then $\cos(\nu \pi)<0$ and $-\left(\nu + \frac{1}{2}\right)>0$, while $\nu \in \left(-\frac{1}{2},\frac{1}{2}\right)$ gives $\cos(\nu \pi)>0$ and $ -\left(\nu + \frac{1}{2}\right)<0$. We conclude that $\gamma_l < 0$ for $M<1$, and we get
\begin{equation}
\braket{u_{l,m},u_{p,s}} = \delta_{lp} \delta_{ms} ,\quad M<1.
\end{equation}
If $M>1$, let $L = \pm \sqrt{M^2 - 1}$. Hence, \eqref{eq:gamma_zero_final} gives
\begin{equation}
\frac{1}{\gamma_0} = \frac{L}{\pi} \sinh (L \pi),
\end{equation}
then $\gamma_0$ is a positive function on the variable $\lambda$, therefore $\gamma_l>0$ and consequently
\begin{equation}
\braket{u_{l,m},u_{p,s}} = \delta_{lp} \delta_{ms} ,\quad M>1.
\end{equation}
Remaining identities $\braket{v_{l,m},v_{p,s}} = - \delta_{lp} \delta_{ms}$ and $\braket{u_{l,m},v_{p,s}} = 0$ are straightforward consequences of $\braket{f^*,g^*} = -\braket{f,g}$ and $v_{l,m}^* = u_{l,-m}$.
\end{proof}

%% file: 2_quantization.tex
The Hilbert space $\mathcal{H}$ generated by the completion of $\mathcal{S}^+$ is the so-called \emph{one-particle space} is endowed with a form \eqref{eq:prod_int}, which becomes an inner product in this space. The vectors $\phi \in \mathcal{S}^+$ are superpositions of the positive-energy modes, and they can be written as
\begin{equation}
\phi(t, \theta, \varphi) = \sum_{l=0}^\infty \sum_{m=-l}^l \phi_{l,m} u_{l,m}\text{, } \sum_{l=0}^\infty \sum_{m=-l}^l |\phi_{l,m}|^2 =1, \quad  \phi_{l,m} \in \mathbb{C},
\end{equation}
and for every two vectors $\phi, \psi \in \mathcal{S}^+$, we have $\braket{\phi,\psi}= \sum_{l=0}^\infty \sum_{m=-l}^l \phi_{l,m}^* \psi_{l,m}$. The understanding of this space in this work is no different from the usual physical interpretation: $\phi$ is the wavefunction of a single particle represented in space-time structure. In particular, the structure is the de Sitter space, and as mentioned in \cite{Barata2012}, although the notion of particle is not always well-defined in QFTCS, there are cases where we observe particle-like behavior. Moreover, experiments in particle physics are realized in slightly curved, see \cite{Fewster2008} for details.

The quantization in the de Sitter space is the same one described in \cite{Barata2012}, by an analog construction in the Minkowski space via positive-energy modes, see \cite{Birrel1982}: we consider the bosonic (symmetrized tensor products) Fock space $\mathcal{F} := \mathbb{C} \oplus_{n=1}^\infty \left( \mathcal{H}^{\otimes n} \right)_s$, and we consider the KG Lagrangean for the field $\phi$,
\begin{equation}
\mathcal{L} := \frac{1}{2}g^{\sigma \tau} (\partial_\sigma \phi) (\partial_\tau \phi) - \frac{1}{2} \frac{\mu^2}{\hslash^2} \phi^2 
\end{equation}
where we recall that $\mu = m_p^2 + \xi R$ and we use the Einstein summation above. The conjugate momentum for $\phi$, denoted by $\pi_\phi$, is given by $\pi_\phi := \partiald{(D_g \mathcal{L})}{(\partial_t \phi)} = \partial_t \phi$. Then, the massive neutral scalar quantum field is
\begin{equation}
\label{eq:campo_quantizado}
\hat{\phi} (t, \theta, \varphi) = \sum_{l=0}^\infty \sum_{m=-l}^l ( a_{l,m} u_{l,m} + a_{l,m}^* u_{l,m}^*).
\end{equation}
where $u_{l,m}$ the positive-energy modes, $a_{l,m}$ and $a_{l,m}^*$ are the annihilation and creation operators respectively, and these obey the commutation rules
\begin{equation}
\label{eq:comutacao}
[a_{lm},a_{ps}^*] = \delta_{lp} \delta_{ms} \text{ e } [a_{lm},a_{ps}] = [a_{lm}^*,a_{ps}^*] = 0,
\end{equation}
where the vacuum state $\ket{V}$ is the one satisfying $a_{lm}\ket{V}$ for every $l \geq 0$. $-l \leq m \leq l$. We remark that one could use a different positive mode decomposition that is not equivalent to this one, this is a feature of the quantization on curved spacetimes. Non-equivalent representations may generate different vacuum states. In this paper, we choose the representation that is the tri-dimensional version of \cite{Barata2012}, where the vacuum choice corresponds to the one from Bunch and Davies's work \cite{Bunch1978}. The two-point function $G$ for the quantized free field on $dS^3$ is
\begin{equation}
G := \braket{V|\hat{\phi} (t, \theta, \varphi), \hat{\phi} (t',0,0)|V} = \sum_{l=0}^\infty \sum_{m=-l}^l u_{l,m}(t,\theta,\varphi)u_{l,m}(t',0,0)^*.
\end{equation}

\subsection{Action of the isometry group on the positive-energy solutions.}

Now we show that the space $\mathcal{S}^+$ is invariant under actions of $O(3,1)$, that is, it does not depend on the coordinates, and the method is to verify how rotations, Lorentz boosts and the parity and time-reversal operators acts on the elements $u_{l,m}$, since every element of $O(3,1)$ is a composition involving the mentioned operators. An immediate consequence of the invariance is the fact that the completion of $\mathcal{S}^+$, namely $\mathcal{H}$, is an irreducible representation of the de Sitter group. 

Starting with the discrete symmetries, by denoting $P_i$, $i = 1,2,3$, as the parity operator associated with the coordinate $X^i$, and $P = P_1 P_2 P_3$, we obtain via usual relations on the spherical harmonics the following:
\begin{align}
 P_1 u_{l,m} (t, \theta, \varphi) &= u_{l,m} (t, \theta, \pi-\varphi) = \frac{(l+m)!}{(l-m)!} u_{l,-m} (t, \theta, \varphi),\\
 P_2 u_{l,m} (t, \theta, \varphi) &= u_{l,m} (t, \theta, -\varphi) = (-1)^m \frac{(l+m)!}{(l-m)!} u_{l,-m} (t, \theta, \varphi),\\
 P_3 u_{l,m} (t, \theta, \varphi) &= u_{l,m} (t, \pi-\theta, \varphi) = (-1)^{l-|m|} u_{l,m} (t, \theta, \varphi) \label{eq:P3}, \\
 P u_{l,m} (t, \theta, \varphi) &= (-1)^l u_{l,m} (t, \theta, \varphi). \label{eq:P}
\end{align}
Now, the time-reversal operator, denoted by $T$, acts on $u_{l,m}$ as follows: 
\begin{equation}
 Tu_{l,m} (t, \theta, \varphi) = [u_{l,m} (-t, \theta, \varphi)]^* = (-1)^m u_{l,m} (t, \theta, \varphi).
\end{equation}

Now, concerning the continuous symmetries, we present how rotations and Lorentz boosts act on $u_{l,m}$. Let $\SRot$ be the subgroup of spatial rotations on $M^4$, and $U_{ab}(\psi)$, $a,b = 1,2,3$, $a\neq b$, be a counterclockwise rotation normal to the plane $X^a \times X^b$ (usual orientation) by an angle $\psi$, which we will use later. The generators of the rotations are given by
\begin{equation}
 N_{ij} = -\sum_{k=1}^3 \epsilon_{ijk}X^i\partiald{}{X^j},
\end{equation}
where $\epsilon_{ijk}$ is the Levi-Civita symbol. By applying these operators on $u_{l,m}$, we obtain
\begin{align}
  N_{12}u_{l,m} &= -imu_{l,m}, \\
  N_{23}u_{l,m} &= \frac{i}{2} u_{l,m+1} + \frac{i}{2} [l(l+1) - m(m-1)] u_{l,m-1}, \\
  N_{31}u_{l,m} &= \frac{1}{2} u_{l,m+1} - \frac{1}{2} [l(l+1) - m(m-1)] u_{l,m-1} .
\end{align}

It remains to understand the action of the Lorentz boosts on the basis of $\mathcal{S}^+$. The generators of the boosts are the operators $N_{0k}$, $k = 1, 2, 3$, where $k$ corresponds to the direction of axis $X^3$. They are defined below:
\begin{equation}
  N_{0k} = X^k \partiald{}{X^0} + X^0 \partiald{}{X^k}.
\end{equation}
The actions of the operators above on the positive modes basis is 
\begin{align}
  N_{01}u_{l,m} &= i \frac{\sqrt{\left(-\nu -l -\frac{3}{2} \right)\left(\nu -l -\frac{1}{2} \right)}}{2(2l+1)} \big[(l-m+1)(l-m+2) u_{l+1,m-1} - u_{l+1,m+1} \big] \nonumber\\
  & -i \frac{\sqrt{\left(-\nu -l -\frac{1}{2} \right)\left(\nu -l +\frac{1}{2} \right)}}{2(2l+1)} \big[(l+m)(l+m-1) u_{l-1,m-1} - u_{l-1,m+1}\big],\\
  N_{02}u_{l,m} &= - \frac{\sqrt{\left(-\nu -l -\frac{3}{2} \right)\left(\nu -l -\frac{1}{2} \right)}}{2(2l+1)} \big[(l-m+1)(l-m+2) u_{l+1,m-1} + u_{l+1,m+1} \big] \nonumber\\
  & + \frac{\sqrt{\left(-\nu -l -\frac{1}{2} \right)\left(\nu -l +\frac{1}{2} \right)}}{2(2l+1)} \big[(l+m)(l+m-1) u_{l-1,m-1} + u_{l-1,m+1}\big], \\
  N_{03}u_{l,m} &=  i \frac{\sqrt{\left(-\nu -l -\frac{3}{2} \right)\left(\nu -l -\frac{1}{2} \right)}}{2l+1} (l-m+1)u_{l+1,m} \nonumber\\
  & +i \frac{\sqrt{\left(-\nu -l -\frac{1}{2} \right)\left(\nu -l +\frac{1}{2} \right)}}{2l+1} (l+m) u_{l-1,m} .
\end{align}
We conclude that $\mathcal{H}$ is algebraically closed under actions of the symmetry group $O(3,1)$. The Casimir operators that characterize irreducible representations of $O(3,1)_+^\uparrow$ are
\begin{align*}
 Q &= N_{12}^2 + N_{23}^2 + N_{31}^2 - N_{01}^2 - N_{02}^2 - N_{03}^2, \\
 R &=  - N_{01} N_{23} - N_{02} N_{31} - N_{03} N_{12},
\end{align*}
and then
\begin{align*}
 Q = \left(\frac{3}{4} - \nu (\nu+1) \right)1 = \frac{\alpha^2 \mu^2}{\hslash^2} 1 \; \text{ and } \; R =  0. 
\end{align*}
 We recall that the restriction $\mu^2 > 0$ gives $\nu \in \mathbb{R}\cup \left\{-\frac{1}{2}+i \sigma: \sigma \in \mathbb{R}\right\}$. Accordingly to Bargmann \cite{Bargmann1947}, by setting $Q = q1$, the condition $\nu \in \mathbb{R}$ gives $0<q<1$, corresponding the \emph{complementary series}. On the other hand, $\nu \notin \mathbb{R}$ implies $q>1$, corresponding to the \emph{principal series}.

%% file: 3_localizability.tex
\subsection{Postulates for localizability}

As we mentioned in the introduction of this work, the idea of localization of relativistic quantum particles was started in \cite{Newton1949}. They postulated conditions that make it possible to construct a unique position operator (see \cite{haag1996} for a more direct approach). In their work, $M^{1+1}$, the particle subspace is the vector space $L^2 \left( \mathbb{R}, \frac{dp}{\omega(p)} \right)$ endowed with the usual inner product, where $\omega(p):= \sqrt{p^2 + m^2}$, in the momentum representation, and $m$ is the mass of the particle. A unitary transformation that absorbs the factor $\omega(p)$ in the measure in the wavefunctions, that is, $\phi(p) \mapsto \phi_{NM}(p) := \frac{\phi(p)}{\sqrt{\omega(p)}} \in L^2 \left( \mathbb{R}, dp \right)$. The Newton-Wigner wavefunction $phi_{NM}(x,t)$ is the Fourier transform of the time evolution of $\phi_{NM}$, that is
\begin{equation}
\label{eq:Fourier}
\phi_{NW} (x,t) = \frac{1}{\sqrt{2 \pi}} \int dp e^{ipx/\hslash} [U_t(\phi_{NW})](p),
\end{equation}
where $U_t$ is the time evolution operator, $U_t:L^2 \left(\mathbb{R}, dp \right) \to L^2 \left(\mathbb{R}, dp \right)$, $U_t(\phi_{NW})(p) := e^{-i \omega(p)t/\hslash} \phi_{NW}(p)$. The probability density function $P(x,t)$ on finding the particle at position $x$ and time $t$ is $P(x,t) := |\phi_{NW} (x,t)|^2$, and the position operator at time $t$, denoted by $q_t$, is given by
\begin{equation}
\label{eq:posicao_NW}
q_t(\phi_{NW} (x,t)) = x \phi_{NW} (x,t).
\end{equation}

In \cite{Barata2012}, since the spatial dimension is one, the subspace of positive-energy modes corresponds to the was constructed by a Fourier expansion on the spatial part (azimuthal angle). Such expansions were approached as an analog to the Fourier transform for de Sitter space in $1+1$ dimensions. In our case, we use the spherical harmonic expansion in the spatial expansion and consequently, we also avoided the problem of the absence of a canonical definition for the momentum space representation. Another similarity between the last mentioned work and this one is that we carry the problem that the Newton-Wigner position operator is not covariant and hence it depends on the reference frame, and it is not in our purposes to solve it in this article.

To define the localizability, We recall that a unitary transformation between two Hilbert spaces $\mathcal{H}_1$ and $\mathcal{H}_2$ is a transformation $W:\mathcal{H}_1 \to \mathcal{H}_2$ that is an isomorphism and an isometry with respect their inner products, i.e.,
\begin{equation*}
\braket{W(f),W(g)}_{\mathcal{H}_2}=\braket{f,g}_{\mathcal{H}_1}.
\end{equation*}

\begin{definition}[Localizability postulates for the tridimensional de Sitter space]
\label{def:postulados}

A free field $\phi$ on the de Sitter is said to be \emph{localizável} when the following occurs:
\begin{itemize}
\item[$\mathbf{I.}$] There exists a family $\{W_t\}_{t \in \mathbb{R}_+}$ of unitary transformations $W_t:\mathcal{H} \to L^2(S^2)$, $W_t(\phi) := \phi_{NW}(t,\theta, \varphi)$, where $L^2(S^2)$ is the Hilbert space of squared-integrable function on the 2-sphere $S^2$. In addition, for each $t \geq 0$, $W_t^{-1}$ depends continuously on the particle mass $m_p$ from equation \eqref{eq:KG};
\item[$\mathbf{II.}$] for a rotation $U(\lambda) \in$ $\SRot$ by an angle $\lambda\in [0,2\pi]$ with respect to any of the three cartesian axes, we have $W_t \circ U(\lambda) = R(\lambda) \circ W_t$, where $R(\lambda)$ is a rotation by an angle $\lambda$ with respect to the same axis and direction of $U$, but acting on em $L^2(S^2)$;
\item[$\mathbf{III.}$] the parity transformation and time-reversal are carried by $W_t$ as follows:
\begin{align*}
P_1 \phi(t, \theta, \varphi) &\to \phi_{NW}(t,\theta, \pi - \varphi),\\
P_2 \phi(t, \theta, \varphi) &\to \phi_{NW}(t,\theta, - \varphi),\\
P_3 \phi(t, \theta, \varphi) &\to \phi_{NW}(t,\pi -\theta,  \varphi),\\
T \phi(t, \theta, \varphi) &\to \phi_{NW}(-t,\theta,\varphi)^*;
\end{align*}
\item[$\mathbf{IV.}$] we have $\phi_{NW}(t, \theta, \varphi) \propto \phi(t, \theta, \varphi)$ in the large mass limit.
\end{itemize} 
\end{definition}

The postulates above are the $3$-dimensional version of the notion of localization system presented in \cite{Barata2012}, and they carry analog intuitive contents as in the $2$-dimensional version that we briefly recall here. In short, postulate I ensures the probabilistic interpretation of the wavefunction $\phi_{NM} \in L^2(S^2)$ for every time $t$, and the continuity of $W_t^{-1}$ is a regularity condition. Postulate II imposes the well-behavior of the Newton-Wigner representation under rotations, that is, rotations on the de Sitter group also rotate the probability amplitudes on the $2$-sphere by the same polar and azimuthal angles. Now, postulate III establishes that discrete symmetries in $\mathcal{H}$ are carried as geometrical transformations in the NW representation. The second Wigner theorem states that quantum symmetries are in general defined up to a phase, and in this work, we set the phase to $1$ so we do not need to deal with complications associated with possible projective representation of the extended de Sitter group. Finally, postulate IV is used to eliminate ambiguities when we determine the transformations $\{W_t\}_{t \in \mathbb{R}_+}$ later in this paper. In the same way as in \cite{Barata2012}, the motivation for this last postulate comes from the fact that the Hermitian form \eqref{eq:prod_int} becomes an inner product in $L^2(S)$ in the large mass limits. In other words,
\begin{equation}
\label{eq:prop_produtos}
\braket{f,g} \simeq 2 \frac{\mu [a(t)]^2}{\hslash} \int f(t,\theta, \varphi)^* g(t, \theta, \varphi) sin \theta d \theta d \phi,
\end{equation}
see appendix for details. Then, it is intuitive to understand $|\phi(t,\theta,\varphi)|^2$ as a probability distribution.

\subsection{Construction of the unitary transformations for the localizability system on the tridimensional de Sitter space}

Now we construct the family $\{W_t\}_{t \in \mathbb{R}_+}$ as in Definition \ref{def:postulados} for the system of a bosonic particle from its positive-energy solutions by applying the localizability postulates.

For every $t \in \mathbb{R}$, there exists a basis $\mathfrak{B}$ of $L^2(S^2)$ whose elements are eigenvectors of the Hermitian generator of the rotations $J:= i N_{12}$ for the axis $X^3$. Such a basis is the set of spherical harmonic functions. Then, we may write the Newton-Wigner function $\phi_{NW} = W_t(\phi)$ of a Klein-Gordon field $\phi$ as an expansion in that basis. Hence,
\begin{equation*}
\phi_{NW} = \sum_{l=0}^\infty \sum_{m=-l}^l q_{lm} (t) Y_l^m(\theta, \varphi) \text{, where } \sum_{l=0}^\infty \sum_{m=-l}^l |q_{lm} (t)|^2 = 1 \text{ and } q_{lm}(t) \in \mathbb{C}.
\end{equation*}
By postulates I and II, we can write
\begin{equation}
\label{eq:rot_S2}
W_t \circ U_{ij} (\alpha) = R_{ij}(\alpha) \circ W_t, \text{ }i,j = 1,2,3, i\neq j, 
\end{equation}
where $\alpha$ is the rotation angle. Each $U_{ij}$ is linear, and so is for the image on $L^2(S^2)$. For rotations on $X^3$, we have
\begin{equation*}
W_t(U_{12}(\alpha)u_{l,m}) =  W_t(e^{-im\alpha} u_{l,m}) = e^{-im\alpha} W_t(u_{l,m})
\end{equation*}
and \eqref{eq:rot_S2} gives
\begin{equation}
\label{eq:autovet}
R_{12}(\alpha) W_t(u_{l,m}) = e^{-im\alpha} W_t(u_{l,m})
\end{equation}
Then, $W_t(u_{l,m})$ is an eigenvector for $R_{12}(\alpha)$ with associated eigenvalue $e^{-im\alpha}$.
Observe that $J$ is non-degenerate and for each $l \in \mathbb{N}_0$, a natural candidate for $W_t (u_{lm})$ is given by
\begin{equation}
\label{eq:candidato}
W_t(u_{lm}(t,\theta, \varphi)) = s_{l,m}(t) Y_l^m(\theta, \varphi), \text{ }b_{l,m}(t) \in \mathbb{C}.
\end{equation}
We show now that the above is the solution. In fact, consider a generic rotation by the Euler angles $\xi$ (axis $X^3$), $\epsilon$ (axis $X^2$) and $\tau$ (axis $X^3$ again), all of them counter-clockwise, denoted by $R(\xi,\epsilon,\tau)$. Its action on the spherical harmonics is
\begin{equation}
\label{eq:rot_geral}
R(\xi,\epsilon,\tau) Y_l^m (\theta, \varphi) = \sum_{k=-l}^l D_{km}^{(l)}(\xi,\epsilon,\tau) Y_l^k(\theta, \varphi),
\end{equation}
where $D_{km}^{(l)} \in \mathbb{C}$, $k \in {-l, -l+1, \cdots, l-1, l}$, is the element of the Wigner matrix \cite{Joshi1977}. By identity \eqref{eq:rot_geral}, for each $l \in \mathbb{N}_0$, the space $\Upsilon_\ell := \spann \{Y_l^m: m \in \mathbb{Z}, |m| \leq l\}$ is invariant under rotations, and the postulate II imposes that the dimension of the space of solutions of \eqref{eq:autovet} is $1$, therefore the solution must be \eqref{eq:candidato}. Moreover, since $W_t$ is unitary, we have
\begin{equation*}
\braket{W_t(u_{l,m}),W_t(u_{l,m})} = |s_{l,m}(t)|^2 = 1,
\end{equation*}
that is, $s_{l,m}(t) = e^{-i\zeta_{l,m}(t)}$, a time-dependent phase, and hence
\begin{equation}
\label{eq:W_parcial}
W_t(u_{l,m}(t,\theta, \varphi)) = e^{-i\zeta_{l,m}(t)} Y_l^m(\theta, \varphi).
\end{equation}
Now, by postulate III we have
\begin{eqnarray}
W_t(u_{l,m}) &=& W_t \left( \frac{(l+m)!}{(l-m)!} P_1  u_{l,-m} \right) = e^{-i \zeta_{l,-m}(t)} Y_l^m (\theta, \varphi), \label{eq:projW}\\
W_t(u_{l,m}) &=& W_t((-1)^m T u_{l,-m}) = e^{i \zeta_{l,-m}(-t)} Y_l^m (\theta, \varphi), \label{eq:timeW}
\end{eqnarray}
and the exact same result in \eqref{eq:projW} occurs for $P_2$. The parity operators $P_3$ and $P$ do not generate results with different indices (see \eqref{eq:P3} and \eqref{eq:P}) and hence they bring no new conditions for $W_t(u_{l,m})$. Identities \eqref{eq:projW} and \eqref{eq:timeW} give
\begin{equation}
\label{eq:condW}
s_{l,m} (t) = s_{l,-m} (t) \text{ and } s_{l,m}(t) = s_{l,-m}(-t)^*
\end{equation}
So far postulates I-III restrict significantly the transformations $W_t$, however, they are still arbitrary with respect to the phase $\zeta_{l,m}$. Analogously to \cite{Barata2012,Newton1949}, we determine $W_0$ and then show that the postulate IV is a sufficient condition for the uniqueness, and $W_t$ is determined via time evolution. For $t=0$, we have
\begin{equation}\label{eq:s_l_m_zero_ambiguous}
s_{l,m} (0) = s_{l,-m} (0) = \pm 1.
\end{equation}
The transformation for $t=0$ is given by
\begin{equation}
\label{eq:W0}
\phi (0, \theta, \varphi) = \sum_{l,m} \phi_{l,m} \sqrt{\frac{\gamma_l}{2 \alpha}}(-i) T_\nu^{l+ \frac{1}{2}}(0) Y_l^m(\theta, \varphi) \mapsto \phi_{NW} (0,\theta, \varphi) = \sum_{l,m} \phi_{l,m} s_{l,m}(0) Y_l^m(\theta, \varphi),
\end{equation}
where, similarly to \cite{Barata2012} we have ambiguity of signal for each term of the series in the LHS. Such an ambiguity appeared for the first time in \cite{philips1968}. The two solutions for $\nu$ in \eqref{eq:nu} imply that the asymptotic limits for large masses are given by $\nu \simeq -\frac{1}{2} \pm iM$, where $M =  \frac{\alpha \mu}{\hslash} \simeq \frac{\alpha m_p}{\hslash}$. To satisfy postulate IV, we necessarily have $\phi_{NW}(0,\theta,\varphi) = f(m_p)\phi(0,\theta,\varphi)$, where $f(m_p)$ is a normalization factor that depends on the particle mass. One can show that
\begin{equation}
\label{eq:T0}
T_\nu^\sigma (0) = \frac{\sqrt{\pi}}{2^\sigma} \frac{\Gamma(\nu + \sigma +1)}{\Gamma(\nu - \sigma +1) \Gamma(\frac{\sigma - \nu +1}{2}) \Gamma(\frac{\nu + \sigma}{2} +1)}, \text{ }\sigma \in \mathbb{C},
\end{equation}
and by \eqref{eq:T0} and the Stirling approximation formula, we obtain
\begin{equation}\label{simeq:postulate_4}
\sqrt{\frac{\gamma_l}{2 \alpha}} N(\nu) T_\nu^{l+ \frac{1}{2}}(0) \simeq (-1)^l (2 \alpha M)^{-\frac{1}{2}}.
\end{equation}
The term $(2 \alpha M)^{-\frac{1}{2}}$ is a coefficient that depends on the particle mass and it normalizes the Hermitian form \eqref{eq:prop_produtos} for $a(0) = \alpha$ in the limit of large mass. The approximation \eqref{simeq:postulate_4} gives $s_{l,m} (0) = (-1)^l$, and therefore the postulates I-IV determines uniquely $W_0$. Observe that $s_{l,m}$ does not depend on $m$, and hence the invariant spaces $\Upsilon_l$ fix their sign for fixed $l$, and from this point we write $s_{l,m} = s_l$. Since postulate I also states that $W_0^{-1}$ be continuous with respect to the particle mass, if the sign in $s_l$ would change, then we would necessarily have a discontinuity. A heuristic discussion concerning such a signal ambiguity is made in subsection \ref{heuristica}.

Now, we determine the time evolution of $\phi_{NM}$ and, consequently, we obtain $W_t$. Although the four postulates in Definition \ref{def:postulados} fix uniquely $W_0$, we remark that the time-evolution is not unique, and in this paper, we discuss one of these solutions, namely the one analog to the Minkowski space \cite{Newton1949} and the bi-dimensional de Sitter space \cite{Barata2012}. 

Considering the positive-energy modes, any vector $\phi \in \mathcal{H}$ can be written as
\begin{equation}\label{eq:time_evolution_NW_function}
\phi(t,\theta,\varphi) = \sum_{l=0}^\infty \sum_{m=-l}^l \phi_{l,m}\sqrt{\frac{\gamma_l}{2 \alpha}}N(\nu) \frac{T_\nu^{l+\frac{1}{2}}(i\sinh (t/\alpha))}{\sqrt{\cosh (t/\alpha)}}Y_l^m(\theta, \varphi),
\end{equation}
where $\sum_{l,m} |\phi_{l,m}| = 1$. By the Hermitian form \eqref{eq:prod_int}, one gets $\braket{\phi,\psi} = \sum_{l,m} \phi_{l,m}^* \psi_{l,m}$. Define $\Xi_l (t)$ and $\zeta_l(t)$, $l \in \mathbb{N}_0$ as
\begin{equation}
\label{eq:fator_absorvido}
\Xi_l (t) := \frac{\alpha \cosh (t/\alpha)}{\gamma_l \left|T_\nu^{l+\frac{1}{2}}(i\sinh(t/\alpha))\right|^2}
\end{equation}
and
\begin{equation}
\label{eq:fase}
\zeta_l (t) := -\arg \left( N(\nu) \frac{T_\nu^{l+\frac{1}{2}}(i \sinh (t/\alpha)}{\sqrt{\cosh(t/\alpha)}} \right).
\end{equation}
Now define the time-dependent unitary transformation $W_t$ by
\begin{equation}
\label{eq:transf}
\phi (t,\theta,\varphi) \to \phi_{NW} (t,\theta,\varphi) = \sum_{l=0}^\infty \sum_{m=-l}^l \phi_{l,m} e^{-i \zeta_l (t)} Y_l^m(\theta,\varphi),
\end{equation}
i.e., $W_t$ absorbs $[2\Xi_l(t)]^{-\frac{1}{2}}$ on each coefficient $\phi_{l,m}$, a consequence of the postulates I-III, and includes the phase $e^{-i \zeta_l (t)}$ which the time-evolution of each energy mode.

We verify the validity of the postulates on $W_t$. By construction, the solution \eqref{eq:time_evolution_NW_function} satisfies postulate I, and partially II and III. For II and III, it remains to prove that the phase factor satisfies the two identities in \eqref{eq:condW} for every $t$. The first one is straightforward, via the independence between the phase and the order of the spherical harmonics. The second one is proved by using the identity $[\pm iT_\nu^{l+\frac{1}{2}} (i\sinh (t/\alpha))]^* = \pm iT_\nu^{l+\frac{1}{2}} (i\sinh (t/\alpha))$. To prove postulate IV, first, we observe that the recurrence property of the gamma function gives $e^{-i \zeta_l (t)} = (-1)^l$.

For $t \neq 0$, the Stirling approximation on the large mass limit gives
\begin{equation}
\label{eq:gamma_final}
\sqrt{\frac{\gamma_l}{2 \alpha}} = \frac{|\Gamma(-l + iM)|}{\sqrt{2 \alpha}} \simeq \sqrt{\frac{\pi}{\alpha}} M^{-l-\frac{1}{2}}e^{-\pi M/2}.
\end{equation}
The asymptotic formula for $P_\nu^{l+\frac{1}{2}}(z)$,
\begin{equation*}
P_\nu^\sigma (z) \simeq \frac{\nu^{\sigma-1/2}}{\sqrt{2 \pi} (z^2 - 1)^{1/4}} \left[-e^{\pm i \left(\sigma - 1/2 \right) \pi} e^{i\left(\nu + 1/2 \right) \omega} + e^{-i\left(\nu + 1/2 \right) \omega}\right], \text{ }\nu_2 \to \pm \infty, \text{ }z = \cos \omega.
\end{equation*}
gives for $t>0$ (observe that $\omega \simeq -\frac{\pi}{2}+i\frac{t}{\alpha}$ for the Riemann surface chosen for the asymptotic limit) the following:
\begin{equation}
\label{eq:T_assintotico}
\frac{T_\nu^{l+\frac{1}{2}}(i\sinh (t/\alpha))}{\sqrt{\cosh(t/\alpha)}} \simeq \frac{(-1)^lM^l e^{\pi M/2}e^{-iMt/\alpha}}{\sqrt{2\pi}\cosh(t/\alpha)}.
\end{equation}
For $t<0$ the approximation \eqref{eq:T_assintotico} is the same except by an extra factor $i$, however, we always may refix the initial time in $t=0$ and evolve only for increasing time. Hence, \eqref{eq:gamma_final} and \eqref{eq:T_assintotico} give
\begin{equation}
\label{T_final}
\sqrt{\frac{\gamma_l}{2 \alpha}} \frac{T_\nu^{l+\frac{1}{2}}(i\sinh (t/\alpha))}{\sqrt{\cosh(t/\alpha)}} \simeq (-1)^l (2\alpha M [\cosh(t/\alpha)]^2)^{-\frac{1}{2}} e^{-iMt/\alpha}
\end{equation}
where $(2\alpha M [\cosh(t/\alpha)]^2)^{\frac{1}{2}}$ corresponds to a mass-dependent normalization coefficient, present in \eqref{eq:prop_produtos}.

Since the function $\phi_{NW}$ in \eqref{eq:transf} belongs to $L^2(S^2)$, its physical meaning corresponds to $|\phi_{NW}|^2$ be a probability density on $S^2$ on finding the particle. However, the actual radius of the model is the time slice at time $t$, and then it varies with the time, and it is given by the scale factor $a(t) = \alpha \cosh (t/\alpha)$. Therefore, the probability density is obtained by absorbing the factor $a(t)$, hence
\begin{equation}
\label{eq:transformacao_almost}
\tilde{\phi}_{NW} (t, \theta, \varphi) = \frac{1}{\alpha \cosh(t/\alpha)}\sum_{l=0}^\infty \sum_{m=-l}^l \phi_{l,m} e^{-i\zeta_l(t)} Y_l^m (\theta, \varphi).
\end{equation}
And the final absorbed factor $[2\omega_l^{dS}(t)]^{-\frac{1}{2}}$ is
\begin{equation}
\omega_l^{dS}(t):= \frac{1}{\gamma_l \left|T_\nu^{l+\frac{1}{2}}(i\sinh(t/\alpha))\right|^2}.
\end{equation}
Consequently,
\begin{equation}
\label{eq:deriv_zeta2}
\zeta_l'(t) = \frac{1}{\gamma_l a(t) \left|T_\nu^{l+\frac{1}{2}}(i \sinh (t/\alpha) \right|^2} = \frac{\omega_l^{dS} (t)}{a(t)},
\end{equation}
and $\zeta_l'$ differs from the bi-dimensional case \cite{Barata2012} by a factor $\frac{1}{a(t)}$ in $\omega_l^{dS}$. Therefore,
\begin{equation}
\label{eq:objetivo_do_trabalho}
W_t(\phi) := \sum_{l=0}^\infty \sum_{m=-l}^l \sqrt{2\omega_l^{dS}(t)}  \phi_{l,m} \sqrt{\frac{\gamma_l}{2 \alpha}} \frac{N(\nu) T_\nu^{l+\frac{1}{2}}(i\sinh (t/\alpha))}{\sqrt{\cosh (t/\alpha)}}Y_l^m(\theta, \varphi)
\end{equation}
is the transformation that satisfies the localizability conditions for the model of a bosonic particle in the tridimensional de Sitter space. Note that $W_t$ works for both principal and complementary series, but the latter one is non-oscillatory since it is a real series, and then $e^{-i\zeta_l(t)}$ is constant for every $l \in \mathbb{N}_0$, i.e., the dynamics is trivial. We recall that the complementary series is given by $\alpha^2 \mu^2/\hslash^2 < 1$. The physical interpretation for such inequality is the same as in \cite{Barata2012}: the complementary series has the Compton wavelength larger than the de Sitter radius, making the particle confined, while the principal series has the opposite situation. As a consequence, the principal series can be evolved in time for the study of the motion of wavepackets on the semi-classical regime, by analyzing how probability densities of localized wavepackets diffuse through time. 


\subsection{The position operator}

Now we explicit the position operator for the tridimensional de Sitter space, by starting from similar construction done in \cite{Newton1949} and recalled in \eqref{eq:posicao_NW}. We have
\begin{equation*}
\vec{q_t} \phi_{NW} (t,\theta, \varphi) := \hat{r} \phi_{NW} (t,\theta, \varphi),
\end{equation*}
where $\hat{r} = \hat{r} (\theta, \varphi)$ is the unitary vector in $S^2$.
\begin{equation}
\label{eq:versor_HE}
\hat{r} (\theta, \varphi) = \begin{pmatrix}
                        \cos \varphi \sin \theta \\
                        \sin \varphi \sin \theta \\
                        \cos \theta
                        \end{pmatrix}
                        = \begin{pmatrix}
                        \frac{1}{2} \left[ Y_1^1(\theta, \varphi)+ Y_1^{-1}(\theta, \varphi) \right]\\
                        \frac{i}{2} \left[ Y_1^1(\theta, \varphi) - Y_1^{-1}(\theta, \varphi) \right] \\
                        Y_1^0(\theta, \varphi)
                        \end{pmatrix}.
\end{equation}
We denote the $j$-th coordinate of the position operator by $\hat{X}^j_B$, where $B = L^2(S^2)$ or $\mathcal{H}$. The method of obtaining the position operator is to describe $\hat{X}^j_{L^2(S^2)} \phi_{NM}$ in the spherical harmonics basis of $L^2(S^2)$, and then bring it back to $\mathcal{H}$ via $W_t^{-1}$, i.e., $\hat{X}_{\mathcal{H}}^j \phi = W_t^{-1} \left(\hat{X}_{L^2(S^2)}^j \phi_{NW} \right)$. The position operator $\hat{R}$ on $\mathcal{H}$, $\hat{R} = (\hat{X}_{\mathcal{H}}^1, \hat{X}_{\mathcal{H}}^2, \hat{X}_{\mathcal{H}}^3)^T$ is given by
\begin{align}\label{eq:X1} 
 \left[ \hat{X}_{\mathcal{H}}^1 \phi \right](t,\theta, \varphi) &= \sum_{p,l=0}^\infty \sum_{s = -p}^p \sum_{m = -l}^l (-1)^s  \phi_{l,m} e^{-i \zeta_l (t)} \sqrt{\frac{3(2l+1)(2p+1)}{32 \pi \omega_p^{dS}(t)}}  \begin{pmatrix} 
1 & l & p \\
0 & 0 & 0
\end{pmatrix}\\
&\times
\left[ \begin{pmatrix}
1 & l & p \\
1 & m & -s
\end{pmatrix}+
\begin{pmatrix}
1 & l & p \\
-1 & m & -s
\end{pmatrix}\right]\nonumber,
\end{align}

\begin{align}\label{eq:X2} 
 \left[ \hat{X}_{\mathcal{H}}^2 \phi \right](t,\theta, \varphi) &= i \sum_{p,l=0}^\infty \sum_{s = -p}^p \sum_{m = -l}^l (-1)^s  \phi_{l,m} e^{-i \zeta_l (t)} \sqrt{\frac{3(2l+1)(2p+1)}{32 \pi \omega_p^{dS}(t)}}  \begin{pmatrix} 
1 & l & p \\
0 & 0 & 0
\end{pmatrix}\\
&\times
\left[ \begin{pmatrix}
1 & l & p \\
1 & m & -s
\end{pmatrix}-
\begin{pmatrix}
1 & l & p \\
-1 & m & -s
\end{pmatrix}\right] \nonumber,
\end{align}

\begin{align}\label{eq:X3} 
 \left[ \hat{X}_{\mathcal{H}}^3 \phi \right](t,\theta, \varphi) &= \sum_{p,l=0}^\infty \sum_{s = -p}^p \sum_{m = -l}^l (-1)^s  \phi_{l,m} e^{-i \zeta_l (t)} \sqrt{\frac{3(2l+1)(2p+1)}{8 \pi \omega_p^{dS}(t)}}   \begin{pmatrix} 
1 & l & p \\
0 & 0 & 0
\end{pmatrix}\\
&\times
\begin{pmatrix}
1 & l & p \\
0 & m & -s
\end{pmatrix}, \nonumber
\end{align}
where the matrices above are $3$-$j$ Wigner symbols. The uniqueness of $\hat{R}$ is given by the uniqueness of the Newton-Wigner representation.

\subsection{Physical understanding of the signal ambiguity}\label{heuristica}

In a similar fashion done in \cite{Barata2012}, we discuss the two possible solutions in \eqref{eq:s_l_m_zero_ambiguous} before we impose postulate IV. We concluded previously in this paper that the postulates I-III are not sufficient to determine a unique solution for the localizability of the one-particle space via Newton-Wigner representation.

For simplicity, we assume already known that the choice of signal in $s_{l,m}$ does not depend on $m$, i.e., $s_{l,m} = s_l$. From a fixed choice of signals $s_l \in \{-1,1\}$, a change $s_l \mapsto s'_l = (-1)^l$ implies a rotation on the polar angle $\theta \mapsto \pi - \theta$, and it consists into an antipodal reflection, i.e., with respect to the plane $X^1 \times X^2$. It is sufficient to study the signal non-uniqueness via the action of $W_0$ on a large enough set of states chosen with some criteria and from this choice to select signals that avoid the antipodal reflections. As an example for the choice of states, let $l = 0$ (hence $m = 0$) and $l=1$ with $m=1$, with coefficients $a_0 = s_0$ and $a_1 = \frac{s_1}{\sqrt{3}}$ respectively. Then,
\begin{eqnarray}
\phi(0,\theta,\varphi) &=& \sqrt{\frac{\gamma_0}{2 \alpha}} N(\nu)T_\nu^{\frac{1}{2}}(0) s_0 \frac{1}{2}\sqrt{\frac{1}{\pi}} + \sqrt{\frac{\gamma_1}{2 \alpha}} N(\nu)T_\nu^{\frac{3}{2}}(0) s_1 \frac{1}{2}\sqrt{\frac{1}{\pi}}, \\
\phi_{NW} (0,\theta,\varphi) &=& \frac{1+\cos \theta}{2 \sqrt{\pi}}. \label{eq:exampleNW}
\end{eqnarray}
The function \eqref{eq:exampleNW} is monotonically decreasing with $\theta$, with maximum and minimum (zero) points attained at $\theta = 0$ and $\theta = \pi$ respectively. Note these extremal points are antipodal to each other. Then, the state above describes a particle more likely to be found in the positive hemisphere of $X^3$. On the other hand, from \eqref{eq:T0} with $A = \sqrt{1-M^2} \in (0,1)\cup i\mathbb{R}$ gives
\begin{equation*}
N(\nu) T_\nu^{l+\frac{1}{2}}(0) = (-1)^l \frac{\sqrt{\pi}}{2^{l+\frac{1}{2}}} (l+A) \frac{ \Gamma \left(l + A \right)}{\Gamma(l-A)} \frac{1}{\Gamma \left(\frac{l}{2} + 1 + \frac{A}{2} \right) \Gamma \left(\frac{l}{2} + 1 - \frac{A}{2} \right)},
\end{equation*}
and the term $(-1)^l$ is the only factor in $T_\nu^{l + \frac{1}{2}}(0)$ that is not necessarily positive, that is, $N(\nu) T_\nu^{l+\frac{1}{2}}(0)$ has alternating signals on $l$, allowing two options for signals, namely $s_0 = s_1$ and $s_0 \neq s_1$. The choice $s_0 = s_1$ gives
\begin{equation}
\label{eq:antipoda}
\phi (0,\theta,\varphi) \propto \frac{1-\cos \theta}{2 \sqrt{\pi}},
\end{equation}
that is, such a choice gives a term similar to \eqref{eq:exampleNW} except by the exchange $\cos \theta \mapsto -\cos \theta$. The function \eqref{eq:antipoda} increases monotonically with a minimum (zero) and maximum in $\theta = 0$ and $\theta = \pi$. In other words, the region of non-trivial likelihood of finding the particle given by $\phi$ is the antipodal region of higher probability corresponding to $\phi_{NM}$. The conclusion is straightforward: we need to choose $s_0 \neq s_1$ in order to keep the probability distributions coherent, in the sense that the unitary transformations $W_t$ must preserve the region on which $\phi$ is concentrated, which also happens in \cite{Barata2012} for the bi-dimensional case. In addition, observe that the example above is not the only combination that works for $l = 0$ and $l=1$, and the cases may include the variable $\varphi$ and the analysis must be done for each $\varphi \in [0, 2 \pi]$ fixed, the example was chosen by its simplicity.

From another perspective, we may choose the signals via a `maximal localization' condition of the position on the eigenstates: let $\{\delta_{NW}^L\}_{L \in \mathbb{N}_0}$ be a sequence of functions $\delta_{NW}^L: S^2 \to \mathcal{C}$ given by
\begin{equation}
\label{eq:seq_delta}
\delta_{NW}^L (\theta - \theta_0, \varphi - \varphi_0) := \sum_{l=0}^L \sum_{m=-l}^l Y_l^m(\theta_0,\varphi_0)^* Y_l^m (\theta, \varphi), \text{ }\theta_0 \in [0,\pi], \text{ }\varphi_0 \in [0,2\pi].
\end{equation}
The sequence \eqref{eq:seq_delta} converge in the distributional sense for $L \to \infty$ to the Dirac distribution $\delta_{NW}(\theta, \varphi)$ on $S^2$, i.e.,
\begin{equation}
\label{eq:Dirac_delta}
\delta_{NW}(\theta - \theta_0, \varphi-\varphi_0) := \sum_{l=0}^\infty \sum_{m=-l}^l Y_l^m(\theta_0,\varphi_0)^* Y_l^m (\theta, \varphi) = \delta(\cos \theta - \cos \theta_0) \delta (\varphi - \varphi_0),
\end{equation}
where for $\theta_0 \in \{0,\pi\}$, the series above does not depend on $\varphi$. In this case, the sequence simply omits the term $\delta (\varphi - \varphi_0)$. For $\theta_0 = 0$, we have
\begin{equation}
\label{eq:seq_delta_0}
\delta_{NW} (\theta, \varphi) = \sum_{l=0}^\infty Y_l^0(0,0)^* Y_l^0 (\theta, \varphi) = \sum_{l=0}^\infty \left( \frac{2l+1}{4\pi} \right) P_l (\cos \theta) = \delta (\cos \theta - 1),
\end{equation}
where $P_l \equiv P_l^0$. By allowing freedom for the coefficients $S_l(0)$, the functions \eqref{eq:seq_delta}, when $\theta = \theta_0 = 0$, become
\begin{equation}
\delta^L(0,0, \varphi) = \sum_{l=0}^L s_l(0) \sqrt{\frac{\gamma_0}{2 \alpha}} N(\nu)T_\nu^{l+\frac{1}{2}}(0) \left( \frac{2l+1}{4\pi} \right)^{\frac{1}{2}} Y_l^0 (0, \varphi).
\end{equation}
Since $\sgn \left( N(\nu)T_\nu^{l+\frac{1}{2}}(0)\right) = (-1)^l$, the choice $s_l(0) = (-1)^l$ maximizes $|\delta^L(0,0, \varphi)|$ for every $L$, and then $\delta_{NW} (\theta, \varphi)$ is the most concentrated as possible for $\theta = 0$ em $t=0$, obeying the de Sitter group symmetries.

An important difference between the localizability in the bi-dimensional \cite{Barata2012} and the tri-dimensional de Sitter case concerning the choice of the signals is the geometrical nature of the problem: in \cite{Barata2012}, the different choices of the signal correspond to diametrically opposite points in $S^1$, while in this work the different positions in $S^2$ are reflections of each other with respect to the axis $X^3$. Of course, this is a consequence of the choices do not depend on the order of the harmonic spherical functions, which may point out that the choice of signals fixes the marginal probabilities by integration over the azimuthal angle $\varphi$.